\documentclass[12pt]{article}
\usepackage{amsthm,amsfonts,amsmath, amscd, bbold}
\usepackage{mathrsfs}
\usepackage{accents}
\usepackage{authblk}

                                \usepackage{verbatim}
\swapnumbers
                              
\theoremstyle{plain}

\usepackage{enumitem}
\numberwithin{equation}{section}
\newtheorem{theorem}{Theorem}[section]

\newtheorem{lemma}[theorem]{Lemma}
\newtheorem{corollary}[theorem]{Corollary}

\theoremstyle{definition}
\newtheorem{definition}[theorem]{Definition}
\newtheorem{example}[theorem]{Example}

\theoremstyle{remark}
\newtheorem{remark}[theorem]{Remark}
\numberwithin{equation}{section}

\newcommand{\R}{{\mathbb R}}
\newcommand{\bR}{{\mathbb R}}

\newcommand{\N}{{\mathbb N}}

\newcommand{\C}{{\mathbb C}}

\newcommand{\cB}{{\mathcal B}}

\newcommand{\cH}{{\mathcal H}}

\newcommand{\cN}{{\mathcal N}}
\newcommand{\cM}{{\mathcal M}}

\newcommand{\cR}{{\mathcal R}}

\newcommand{\cK}{{\mathcal K}}

\newcommand{\one}{{\bf 1}}

\renewcommand{\Re}{\,\mathrm{Re}\,}   
\renewcommand{\Im}{\,\mathrm{Im}\,}   

\newcommand{\Tr}{\mathrm{Tr}}

\newcommand{\be}{\begin{equation}}
\newcommand{\ee}{\end{equation}}
\newcommand{\bea}{\begin{eqnarray}}
\newcommand{\eea}{\end{eqnarray}}
\newcommand{\beann}{\begin{eqnarray*}}
\newcommand{\eeann}{\end{eqnarray*}}

\newcommand{\sE}{{\mathord{\mathscr E}}}

\newcommand{\sR}{{\mathord{\mathscr R}}}

\usepackage{color}

\begin{document}

\title{Recovery and the Data Processing Inequality for quasi-entropies}

\author[1]{Eric A. Carlen}
\author[2]{Anna Vershynina}
\affil[1]{\small{Department of Mathematics, Hill Center, Rutgers University, 110 Frelinghuysen Road,
Piscataway, NJ 08854-8019, USA}}
\affil[2]{\small{Department of Mathematics, Philip Guthrie Hoffman Hall, University of Houston, 
3551 Cullen Blvd., Houston, TX 77204-3008, USA}}
\renewcommand\Authands{ and }
\renewcommand\Affilfont{\itshape\small}

\date{ }

\maketitle

\begin{abstract}  We prove number of quantitative stability bounds for the cases of equality in Petz's monotonicity theorem for quasi-relative entropies $S_f(\rho||\sigma)$ defined in terms of an operator monotone decreasing functions $f$; and in particular, the R\'enyi relative entropies. Included in our results are bounds in terms of the Petz recovery map, but we obtain more general results. The present treatment is entirely elementary and developed in the context of finite dimensional von Neumann algebras where the results are already non-trivial and of interest in quantum information theory. 
\end{abstract}

\section{Introduction}

Quantum relative entropy was defined first by Umegaki in 1962 \cite{U54} as a formal generalization of the classical relative entropy, also known as Kullback-Leibler divergence \cite{KL51}. For two density matrices $\rho$ and $\sigma$ on a finite-dimensional Hilbert space $\cH$, the quantum relative entropy of $\rho$ with respect to $\sigma$  is defined as
\begin{equation}\label{Um}
S(\rho||\sigma)=\Tr(\rho\log\rho-\rho\log\sigma),
\end{equation}
if the null space of  $\sigma$ is contained in null space of   $\rho$, and $+\infty$ otherwise.
Quantum relative entropy can be viewed as a  measure of distinguishability of two states in terms of large 
deviations theory for repeated measurements \cite{BD08,HP91}. If two states are the same, the quantum relative 
entropy between them is zero, while the relative entropy between mutually singular states is infinity. The  larger the 
value of the relative entropy, the easier it is to distinguish between two quantum states, and this is quantified by the results in 
\cite{BD08,HP91}.

The  relative entropy has a monotonicity property that is fundamental to information theory. In the classical case, this is relatively easy to prove, but in the non-commutative setting of quantum information theory the result is much deeper. The quantum version of this monotonicity property is expressed by the inequality
\begin{equation}\label{linmon}
S(\Phi(\rho)\|\Phi(\sigma))\leq S(\rho\|\sigma)
\end{equation}
where $\cH$ and $\cK$ are Hilbert spaces, and $\Phi:\cB(\cH)\to \cB(\cK)$ is a completely positive trace preserving map (CPTP), also known as a {\em quantum channel}.  The inequality \eqref{linmon} is  the {\em Data Processing Inequality} (DPI). It was proved by 
 Lindblad \cite{L75}, building on work of Lieb and Ruskai \cite{LR73}. 
 The inequality states that the relative entropy can not increase after states pass through a noisy quantum channel (CPTP map); after passing through such a channel, the states become harder to distinguish.

Petz \cite{P86, P88} classified all states that lead to equality in DPI, showing that, for a given CPTP map $\Phi$, two states $\rho$ and $\sigma$ lead to equality in the monotonicity relation if and only if  both $\rho$ and $\sigma$ are perfectly recovered  by a map $\sR_\rho$, now known as a Petz recovery map. That is, equality holds in
\eqref{linmon} if and only if 
 $\sR_\rho(\Phi(\rho))=\rho$ and $\sR_\rho(\Phi(\sigma))=\sigma$. In a  case of particular importance, the map 
 $\sR_\rho$ has a very simple explicit form: Suppose that
 $\Phi:\cB(\cH_1\otimes \cH_2) \to \cB(\cH_2)$ is the {\em partial trace} $\Tr_1$ over $\cH_1$. In notation that will be useful in the more general setting discussed below,
 we write $\cM$ for $\cB(\cH_1\otimes \cH_2)$, $\cN$ for  $\cB(\cH_2)$, and for any density matrix $\tau\in \cM$, we write $\tau_\cN$ to denote
 $\Phi(\tau)$, which as explained below, may be regarded as a conditional expectation with respect to the normalized trace of $\tau$ given $\cN$. 
 Then for all $X\in \cB(\cH_2)$
 \begin{equation}\label{pertzmapINT}
 \sR_\rho(X) = \rho^{1/2} (\rho_\cN^{-1/2}X\rho_\cN^{-1/2})\rho^{1/2}\ .
 \end{equation}
It is evident from this formula that $\sR_\rho(\rho_\cN) = \rho$, but that one cannot then expect, in general, that 
$\sR_\rho(\sigma_\cN) = \sigma$. However, as Petz showed \cite{P88}, this occurs if and only if $\sR_\sigma(\rho_\cN) =\rho$. That is, the equality condition is symmetric in $\rho$ and $\sigma$.

These results motivated the following question: If the decrease of relative entropy after states
pass through a quantum channel is small, how well can  these states be recovered? Work on this question accelerated following a 
breakthrough result by Fawzi and Renner in 2015 \cite{FR15}. They proved   that if the  strong subadditivity inequality (SSA)  of 
Lieb and Ruskai  (from which Lindblad derived his monotonicity theorem) is nearly  saturated, then  {\em quantum Markov chain condition}, 
known \cite{H03} to be necessary and sufficient for equality  in SSA is also nearly satisfied, and they gave a precise quantitative 
version of this stability result. 
This result has numerous applications in quantum information theory, e.g. \cite{SBW14, SW14,W14}. Further refinements of the 
monotonicity relation occurred later in, for example, \cite{BT16, BHOS15, CL12, JRSWW15, SFR16, STH16, W15, Z16}. Most of 
these results involve a recovery channel in the lower bound, and even though it often derived from, or 
otherwise related to, the Petz recovery map, in no case is it  the Petz map itself. 

Very recently, Carlen and Vershynina \cite{CV17A} proved a sharp stability result for the DPI directly in terms of the  
Petz recovery map.  This result shows, in quantitative terms,  that if the decrease in the quantum relative entropy is small after two states 
$\rho$ and $\sigma$ pass a quantum channel, then  the  Petz recovery map $\sR_\rho$  approximately  approximately recovers 
$\sigma$ as well as $\rho$.  The precise statement involves the {\em relative modular operator} $\Delta_{\sigma,\rho}$ 
which for density matrices 
$\rho$ and $\sigma$ acting on $\cH$, is the operator action on the Hilbert space $\cB(\cH)$ as follows:
\begin{equation}\label{relmod}
\Delta_{\sigma,\rho}(X) = \sigma X \rho^{-1}\ ,
\end{equation}  for all $X\in \cM$ in the case that  $\rho$ is {\em invertible}. This is the matricial version of an operator introduced in a 
more general von Neumann algebra context 
by Araki \cite{A76}. $\Delta_{\sigma,\rho}$ is evidently a positive operator on $\cM$ (or $M_n(\C)$)  
equipped with the Hilbert-Schmidt inner product $\langle X,Y\rangle_{HS} = \Tr[X^*Y]$.  When $\rho$ is not invertible, $\rho^{-1}$
should be interpreted as the {\em pseudo-inverse} of $\rho$: If $\Sigma(\rho)$ denotes the  spectrum of $\rho$, and 
$P_\lambda$ denotes the corresponding spectral projection, the generalized inverse is
$\rho^+ = \sum_{\lambda\in \Sigma(\rho)\backslash\{0\}}\lambda^{-1} P_\lambda$.   
 It is easy to see that the eigenvalues of $\Delta_{\sigma,\rho}$ are of the form $\mu \lambda^{-1}$ where 
 $\mu$ is an eigenvalue of $\sigma$ and $\lambda$ is a non-zero eigenvalue of $\rho$. (If $\phi$ and $\psi$ are  
 corresponding eigenvectors, then $|\phi\rangle\langle \psi|$ is an eigenvector of $\Delta_{\sigma,\rho}$ with eigenvalue $\mu \lambda^{-1}$.)
In particular, the operator norm of $\Delta_{\sigma,\rho}$ is given by
\begin{equation}
\|\Delta_{\sigma,\rho}\| = \|\sigma\|\|\rho^+\| = \max\{\mu \ :\ \mu\in \Sigma(\sigma)\} \max\{ \lambda^{-1}\ :\ \lambda\in 
 \Sigma(\rho)\backslash\{0\}\} < \infty\ .
\end{equation}
See the discussion below \eqref{modapp} for an explanation of why this formulation of what $\Delta_{\sigma,\rho}$ means in the degenerate case is the relevant one in this context.

For simplicity of notation, and consistency with standard usage, we generally  write $\rho^{-1}$ for $\rho^+$,  
but wherever it appears in this paper, $\rho^{-1}$ is to be interpreted in this manner. This standard usage is 
non unreasonable:  
When there is an orthogonal projection $P$ such that $\Tr[\rho P]= 0$ but 
$\Tr[\sigma P] \neq 0$, there is a obvious simple test for distinguishing $\sigma$ from $\rho$ by making 
repeated observations corresponding to $P$.

With the relative modular operator now introduced, we  state the result of  \cite{CV17A}.  In the setting of finite dimensional von Neumann algebras, take the 
CPTP map $\Phi$ to be the tracial conditional expectation 
from one von Neumann algebra $\cM$ onto a von Neumann sub-algebra $\cN$.  Then 
\begin{equation}\label{CV}
S(\rho\|\sigma)-S(\Phi(\rho)||\Phi(\sigma))\geq \left(\frac{1}{8\pi}\right)^4\|\Delta_{\sigma, \rho}\|^{-2}\|\sR_\rho(\Phi(\sigma)-\sigma)\|_1^4.
\end{equation}
Another version, with a slightly more complicated  constant, has the roles of $\rho$ and $\sigma$ can be reversed on the right-hand side.

The new results in this paper concern bounds of this sort for 
other relative entropies that have proven to be of interest in quantum information theory, and in particular, R\'enyi relative entropies. 
The Umegaki  relative entropy is a particular case of a class of $f$-relative quasi-entropies, defined by Petz \cite{P85} as follows, for an operator monotone decreasing function $f$ on $(0,\infty)$ and invertible states 
\begin{equation}\label{modapp}
S_f(\rho\|\sigma)=\Tr(f(\Delta_{\sigma,\rho})\rho) = \langle \sqrt{\rho}, f(\Delta_{\sigma,\rho}) \sqrt{\rho}\rangle_{\rm HS}\ .
\end{equation}
In the next section we recall the relevant aspects of the theory of of operator monotone functions. For now, the key point to notice is the we only use functions of $\Delta_{\sigma,\rho}$ applied to $\sqrt{\rho}$, which is of course orthogonal in the Hilbert-Schmidt inner product to the projector onto the null space of $\rho$. For this reason, 
with the relative modular operator defined  using the generalized inverse $\rho^+$ in place of $\rho^{-1}$,
\eqref{modapp} is always meaningful even when $\rho$ is degenerate, and moreover, 
taking $f=-\log$ this formula yields  the Umegaki relative entropy. Likewise, in the definition of the Petz recovery map $\rho_\cN^{-1/2}$ should be read as $(\rho_\cN^+)^{1/2}$,

Another important example, R\'enyi relative entropy, $S_\alpha(\rho\|\sigma)=\frac{1}{\alpha-1}\log\Tr(\rho^\alpha\sigma^{1-\alpha})$ for $\alpha>0$, $\alpha\neq1$, involves the power $f$-relative quasi-entropy obtained by taking the power function as $f(t)=t^{1-\alpha}$. It is well-known that R\'enyi relative entropy satisfies the monotonicity relation when $\alpha\in[0,2]$. 

For the class of $f$-relative quasi-entropies  it was shown \cite{LR99, P85, T09} that they satisfy the monotonicity relation for any CPTP map $\Phi$ and an operator convex function $f$
\begin{equation*}
S_f(\Phi(\rho)\|\Phi(\sigma))\leq S_f(\rho\|\sigma).
\end{equation*}
The equality condition in the monotonicity relation for a large class of functions is discussed in \cite{HMPB11}, where authors show that the monotonicity inequality is saturated if and only if Petz map recovers both states perfectly. 

We investigate the question of almost-perfect recoverability for $f$-relative quasi-entropies, and prove analogs of \eqref{CV} for them. For example, for the R\'enyi entropies $S_\alpha$, $\alpha\in (0,1)$, 
we prove in Corollary~\ref{remax} that for any  density matrices $\rho$ and $\sigma$ such that $\sigma_\cN$ is non-degenerate, 
\begin{equation*}
S_\alpha(\rho||\sigma)-S_\alpha(\rho_\cN||\sigma_\cN)\geq
\frac{1}{1-\alpha} \log\left(1 + {\widehat K}_{\alpha,\rho,\sigma} \max\{\| \sR_\rho(\sigma_\cN)-\sigma\|_1, \| \sR_\rho(\rho_\cN)-\rho\|_1\}^{6-2\alpha}\right)\ ,
\end{equation*}
where ${\widehat K}_{\alpha,\rho,\sigma}$ is constant specified in the corollary. Dropping the requirement that 
$\sigma_\cN$ is non-degenerate, we have similar bound with 
$\max\{\| \sR_\rho(\sigma_\cN)-\sigma\|_1, \| \sR_\rho(\rho_\cN)-\rho\|_1\}$  on the right side replaced by $\| \sR_\rho(\sigma_\cN)-\sigma\|_1$, and a different constant.

\section{Operator monotone functions}

 A function $f: (a,b) \to \R$ is {\em operator monotone} if for any pair of  self-adjoint operators 
 $A$ and $B$ on some 
 Hilbert space that have spectrum in $(a,b)$, $f(A) -f(B)$ is positive semidefinite whenever $A - B$ is positive 
 semidefinite.  We say that $f$ is operator monotone decreasing on $(a,b)$ in case $-f$ is 
 operator monotone. (Traditional usage is asymmetric and gives preference to monotone increase.)
 A {\em Pick function} is a function $f$ that is analytic on the upper half plane and has a positive imaginary part. 
 For an open interval $(a,b)\subset \R$, 
$\mathcal{P}_{(a,b)}$ denotes the class of Pick functions that may be analytically continued 
into the lower half plane 
across the interval $(a,b)$  by reflection.  Thus any $f\in \mathcal{P}_{(a,b)}$ is real on $(a,b)$. Moreover, letting 
$u(x,y)$ and $v(x,y)$  denote the real and imaginary parts of $f$, since $\partial v(x,0)/\partial y \geq 0$, the 
Cauchy-Riemann equations say that $f'(x) = \partial u(x,0)/\partial x \geq 0$. In fact, much 
more is true: K.~L\"owner's Theorem of 1934 states that $f$ is operator monotone on $(a,b)$ if and only 
if it is the restriction of a function $f\in \mathcal{P}_{(a,b)}$ to $(a,b)$.   

A function $f$ is {\em operator convex}  on the positive operators in case for all positive semidefinite 
operators $A$ and $B$, and all $\lambda$ in $(0,1)$, $(1-\lambda)f(A) + \lambda f(B) - f((1-\lambda)A + 
\lambda B))$ 
is positive semidefinite, and $f$ is {\em operator concave} in case $-f$ is operator convex. It turns out that 
every operator monotone function is operator concave \cite[Theorem V.2.5]{B97}. Moreover, a 
function that maps  $(0,\infty)$ onto itself is operator monotone if and only if it is 
operator concave. Thus a function $f$ on $(0,\infty)$ 
is operator convex and  operator monotone decreasing  if and only if $-f \in \mathcal{P}_{(0,\infty)}$. 
One example of an operator monotone decreasing function on $(0,\infty)$ is $f_0(x) = -\log x$, 
and a closely related family of examples is given by $f_\alpha(x) = -x^\alpha$, $\alpha\in (0,1]$.  
The theory of Pick functions is reviewed in  the next section. Most important for us is the canonical 
integral representation of Pick functions.

Every  function $f\in {\mathcal P}_{(0,\infty)}$ has a canonical integral representation 
\cite[Chapter II, Theorem I] {Dono}
\begin{equation}\label{low}
f(x) = a x+ b +\int_{0}^\infty \left( \frac{t}{t^2+1}- \frac{1}{t +x }   \right){\rm d}\mu_f(t)\ ,
\end{equation}
where  $a\geq 0$, $b\in\bR$ and $\mu$ is a positive measure on $(0,\infty)$ such that 
${\displaystyle \int_{0}^\infty  \frac{1}{t^2+1}{\rm d}\mu_f(t)<\infty}$.
Conversely, every such function belongs to ${\mathcal P}_{(0,\infty)}$.

\begin{remark} To facilitate comparison with related formulas that may be 
more commonly used in mathematical physics, 
we have changed the sign of $t$ from what it is in \cite{Dono}; this results in a 
corresponding change in \eqref{muform} below. 
\end{remark}

There is a  simple way to determine $a$, $b$ and $\mu$ corresponding to $f$. 
 The following formulas   \cite[Chapter II, p. 24] {Dono} are readily verified. 
\begin{equation}\label{alphabeta}
a = \lim_{y\uparrow\infty}\frac{f(iy)}{iy} \qquad{\rm and}\qquad b = \Re(f(i))\ .
\end{equation}

Next, if one defines the monotone increasing function $\mu(x) := \frac12 \mu(\{x\}) + \mu((-\infty, x))$, 
according to \cite[Chapter II, Lemma 2] {Dono}  we have that
 \begin{equation}\label{muform}
 \mu(x_1) - \mu(x_0) = \lim_{y\downarrow 0} \frac{1}{\pi} \int_{x_0}^{x_1} \Im f(-x+iy){\rm d} x\ .
 \end{equation}
 
\begin{example}\label{logex} Let $f(x)= \log x$.   Then $\Re(\log(i)) = 0$ and 
$\log(iy)/(iy) = (\log y + i\pi/2)/(iy) \to 0$ as $y\uparrow\infty$, so that $a=b = 0$. It is clear from \eqref{muform} that
${\rm d}\mu(x) = \frac{1}{\pi}\lim_{y\downarrow 0}\Im\log(-x + iy){\rm d}x = {\rm d}x$.
Thus, 
\begin{equation}\label{eq:logex}
\log x=\int_{0}^\infty \left( \frac{t}{t^2+1}- \frac{1}{t +x }   \right){\rm d}t 
\end{equation}
Since the integral is readily evaluated, thus verifying the formula, this proves 
that the logarithm function is in fact operator monotone. 

The integral representation \eqref{eq:logex} is equivalent to the familiar representation
$$\log x=\int_{0}^\infty \left( \frac{1}{t+1}- \frac{1}{t +x }   \right){\rm d}t $$
since ${\displaystyle \int_0^\infty \left(\frac{1}{t+1} - \frac{t}{t^2+1}\right){\rm d}t = 0}$. 
\end{example}

\begin{example}\label{powex} Let $f_\alpha(x)= x^\alpha$, $\alpha\in (0,1)$.  Then evidently 
$a = \lim_{y\uparrow\infty}f_\alpha(iy)/(iy) = 0$. Next, $f_\alpha(i) = \cos(\alpha \pi/2) + 
i \sin(\alpha \pi/2)$, and hence 
$b = \sin(\alpha \pi/2)$.
Finally, for $x>0$, $\lim_{y\downarrow 0} \Im f(-x + iy) = x^\alpha\sin(\alpha \pi)$ so that
${\rm d}\mu(x) = \pi^{-1}\sin(\alpha \pi) x^\alpha{\rm d}x$. This yields the representation
\begin{equation}\label{eq:powex}
x^\alpha =  \sin(\alpha \pi/2)  +  \frac{\sin(\alpha \pi)}{\pi} \int_{0}^\infty t^\alpha\left( \frac{t}{t^2+1}- 
\frac{1}{t +x }   \right){\rm d}t 
\end{equation}
Since the integral is readily evaluated, thus verifying the formula, this proves that the  function 
$f_\alpha(x)$ is in fact operator monotone. 

The integral representation \eqref{eq:powex} is equivalent to the familiar representation
$$x^\alpha=\frac{\sin \alpha \pi}{\pi} \int_{0}^\infty t^\alpha\left( \frac{1}{t+1}- \frac{1}{t +x }   \right){\rm d}t $$
by the same sort of calculation made in the previous example.  The merit of the slightly 
more complicated representation \eqref{eq:powex} lies in the simple relation between 
$f_\alpha$ and $a$, $b$  and $\mu$. 
\end{example}

\section{Petz's proof of the DPI for quasi entropies}

Petz \cite{P85} used the properties of operator monotone decreasing functions to generalize Umegaki's relative 
entropy, producing a family of quasi-entropies that share with the original Umegaki relative entropy its 
fundamental monotonicity property, as we now explain. 

Let $\cM$ be a finite dimensional von Neumann algebra, which we  may regard, for some 
$n\in \N$,  as a subalgebra 
of $M_n(\C)$, the von Neumann algebra of $n\times n$ matrices over $\C$.  Let $\rho$ and 
$\sigma$ be two 
density matrices on $\cM$. That is, $\rho$ and $\sigma$ are positive and have unit trace. 
We shall frequently refer to 
density matrices $\rho$  as {\em states} identifying $\rho$ with the positive linear functions 
$X\mapsto \Tr[\rho X]$ on $\cM$. 
As noted above, the Umegaki relative entropy of $\rho$ with respect to $\sigma$, $S(\rho||\sigma)$, 
can be written as
\begin{equation}\label{altform} 
S(\rho||\sigma)=\Tr [ -\log(\Delta_{\sigma,\rho})\rho]  = \langle \rho^{1/2}, 
-\log (\Delta_{\sigma,\rho})\rho^{1/2}\rangle_{HS}\ .
\end{equation}

Let $\cN$ be a von Neumann sub-algebra of $\cM$, and let $\sE_\tau$ be orthogonal projection 
with respect to the Hilbert-Schmidt inner product from $\cM$ onto $\cN$. This turns out to be a 
conditional expectation in the sense of Umegaki \cite{U54}. (See \cite{CV17A} for a detailed and elementary 
discussion of this topic.)  Lindblad proved \cite{L74} a fundamental monotonicity property of the Umegaki 
relative entropy, namely that with $\rho_\cN := \sE_\tau \rho$ and $\sigma_\cN := \sE_\tau \sigma$,
\begin{equation}\label{lindmon}
S(\rho_\cN||\sigma_\cN)  \leq S(\rho||\sigma) \ .
\end{equation}
In a subsequent paper \cite{L75}, he showed, using the Stinespring Dilation Theorem \cite{St55}, that this readily 
implies his more general monotonicity property mentioned above, namely that
\begin{equation}\label{lindmon2}
S(\Phi (\rho)||\Phi( \sigma))  \leq S(\rho ||\sigma) \ ,
\end{equation}
for any CPTP map $\Phi$ from $\cM$ to $\cN$, 
where now it is no longer 
necessary that $\cN$ be a sub-algebra of $\cM$. However, the extension is simple, 
and the main result is contained in \eqref{lindmon}.

Given an operator monotone decreasing function $f$ on $(0,\infty)$, Petz defined the {\em $f$-relative 
quasi-entropy (a.k.a. $f$-divergence) $S_f(\rho||\sigma)$} by 
\begin{equation}\label{Petsqdef}
S_f(\rho||\sigma)=\Tr [ f(\Delta_{\sigma,\rho})\rho]  = \langle \rho^{1/2}, 
f (\Delta_{\sigma,\rho})\rho^{1/2}\rangle_{HS}\ .
\end{equation}
Comparing with \eqref{altform}, we see that the generalization consists of replacing the 
specific operator monotone decreasing function, $-\log$, with a general function of this type. 
Petz then proved \cite[Theorem 4]{P85}
that Lindblad's monotonicity property holds in his more general setting. That is, for all such $f$, 
\begin{equation}\label{lindmon3}
S_f(\Phi (\rho)||\Phi( \sigma) ) \leq S_f(\rho ||\sigma) .
\end{equation}
Again, the essence of the matter is in the case of tracial conditional expectations, 
from which the general cases again follows, and our focus is on the inequality
\begin{equation}\label{lindmon4}
S_f(\rho_\cN||\sigma_\cN)  \leq S_f(\rho||\sigma) \ .
\end{equation}

Various special cases of the quasi-entropies had been considered earlier, for example, the 
R\'enyi relative entropies are closely related to what one obtains from the choice 
$f_\alpha(x) := -x^\alpha$, $\alpha\in (0,1)$, as we discuss below in detail. 

In the rest of the paper, $\|\cdot\|_p$, $p\in [1,\infty)$ denotes that Schatten $p$-trace norm; i.e, 
$\|X\|_p$ is the $\ell_p$ norm of the vectors of singular values of $X$. We simply write 
$\|X\|$ to denote the operator norm of $X$; i.e., the largest singular value of $X$. 

We shall show below that for any $\beta\in (0,1)$ and a very broad class of operator monotone decreasing functions $f$, that depend on parameter $c>0$,
there is an explicitly computable constant $K$ depending only on 
 $\|\rho^{-1}\|$, $\beta$, and $f$ such that for $\beta \leq 1/2$,
\begin{equation}\label{intro1}
\|\sigma_\cN^{\beta}\rho_\cN^{-\beta} \rho^{1/2} - 
\sigma^{\beta}\rho^{\frac{1}{2}-\beta}\|_2  \leq K 
(S_f(\rho||\sigma)-S_f(\rho_\cN\|\sigma_\cN))^{\frac{\beta(1-\beta)}{1+2c(1-\beta)} }\ ,
\end{equation}
while for $\beta \geq 1/2$, 
\begin{equation}\label{intro2}
\|\sigma_\cN^{\beta}\rho_\cN^{-\beta} \rho^{1/2} - 
\sigma^{\beta}\rho^{\frac{1}{2}-\beta}\|_2  \leq K 
(S_f(\rho||\sigma)-S_f(\rho_\cN\|\sigma_\cN))^{\frac12  \frac{1-\beta}{ 1+c} }\ .
\end{equation}
In fact, we prove a more general, if somewhat more complicated, result in Theorem~\ref{thm:f-error}.

The {\em Petz recovery map} $\sR_\rho$ is defined as follows: For all $X\in \cN$, 
 \begin{equation}\label{ac14}
\sR_\rho(X) = \rho^{1/2} (\rho_\cN^{-1/2}X\rho_\cN^{-1/2})\rho^{1/2}\ .
\end{equation}
It is evident form this formula  that $\sR_\rho$ is a CPTP map, and that 
$\sR_\rho(\rho_\cN) = \rho$, which is the reason for the term ``recovery map''. See \cite{CV17A} for an extensive and self-contained discussion of this 
map and the closely related Accardi-Cecchini coarse graining operator \cite{A74,AC82}.  In what follows we use
inequalities of the type \eqref{intro1} and \eqref{intro2} to obtain bounds on 
$$\max\{ \|\sR_\rho(\sigma_\cN)-\sigma\|_1 \ ,\ \|\sR_\sigma(\rho_\cN)-\rho\|_1\} $$
 and
 $\|\sigma_\cN^{\beta}\rho_\cN^{-\beta}  - 
\sigma^{\beta}\rho^{-\beta}\|_2$
in terms of 
$S_f(\rho||\sigma) - S_f(\rho_\cN||\sigma_\cN)$.
In particular, we shall see that for a broad class of operator monotone decreasing functions $f$, not only 
is it the case that any one of them vanishes if and only if the others all vanish, but we can quantitatively relate their sizes.

\section{Stability of $f$-relative quasi-entropy}

In this section we examine the monotonicity property (\ref{lindmon4}) of $f$-relative 
quasi-entropy for a broad class of  operator monotone decreasing functions $f$.

\begin{definition} A function $f \in \mathcal{P}_{(0,\infty)}$ is {\em regular} in case the measure 
$\mu$ in the canonical integral representation of $f$ is absolutely continuous with respect to 
Lebesgue measure, and for each $S,T > 0$, there is a finite constant $C^f_{S,T}$ such that
\begin{equation}\label{regdef}
{\rm d} t \leq C^f_{S,T} {\rm d}\mu(t)
\end{equation}
on the interval $[1/S,T]$.  An operator monotone decreasing function is {\em regular} if and only if $-f$ is regular. 
\end{definition}

The examples from the previous section show that the Pick functions $f_0(x) = \log (x)$ and 
$f_\alpha(x) = x^\alpha$, $\alpha\in (0,1)$, are regular.

For every regular operator monotone decreasing function $f$, we produce a one parameter 
family of lower bounds on 
\begin{equation}\label{focus1}
S_f(\rho||\sigma)-S_f(\rho_\cN\|\sigma_\cN)
\end{equation}
in terms of
\begin{equation}\label{focus2}
\|\sigma_\cN^{\beta}\rho_\cN^{-\beta} \rho^{1/2} - \sigma^{\beta}\rho^{\frac{1}{2}-\beta}\|_2
\end{equation}
for each $\beta\in (0,1)$. The bounds will show, in particular, that the difference in \eqref{focus1} 
vanishes if and only if  
\begin{equation}\label{focus3}
\sigma_\cN^{\beta}\rho_\cN^{-\beta}  =  \sigma^{\beta}\rho^{-\beta}
\end{equation}
for all $\beta\in (0,1)$. Then, since for any strictly positive matrix $X$, $\beta \mapsto X^\beta$ 
is an entire analytic function, \eqref{focus3} is valid for all $\beta\in \C$.   As we discuss later, 
this is closely related to a result of Petz, proved in a more general von Neumann algebra setting 
without assuming any finite dimensionality, but then of course, restricting to purely imaginary 
values of $\beta$, and expressing everything in terms of modular operators. 

The main novelty of our work is that we prove a {\em quantitative relation} between the quantities in 
\eqref{focus1} and \eqref{focus2}, and do not only concern ourselves with cases of equality. 
The reader who is familiar with the Tomita-Takesaki Theory will also see how to generalize a 
number of our results beyond the case in which $\cM$ and $\cN$ are finite dimensional, and 
we plan to return to this in later work. However, the results are new and  interesting already in 
the present context, and it is therefore worthwhile to explain them in their simplest setting, 
which is in any case  the main arena of quantum information theory.

\begin{theorem}\label{thm:f-error}   Let $\cN$ be a von Neumann sub-algebra of a finite dimensional von Neumann algebra $\cM$. 
Let $f$ be a  regular operator monotone decreasing function, and $T>0$. \\
\smallskip
\noindent{\it (1)} For $\beta\leq1/2$, define $T_L(\beta):=T$ and $T_R(\beta):=T^{\beta/(1-\beta)}$.

\smallskip
\noindent{\it (2)} For $\beta\geq1/2$,  define $T_L(\beta):=T^{{(1-\beta)}/{\beta}}$ and $T_R(\beta):=T$.\\
\\
Define
$C_{T, \beta}^f$ to be the least positive constant such that  ${\rm d}t\leq C^f_{T, \beta}{\rm d}\mu_f(t)$ for 
$t\in[T_L(\beta)^{-1}, T_R(\beta)]$, noting that $C_{T, \beta}^f>0$ since $f$ is regular.  Then for all density matrices $\rho$ and $\sigma$ in $\cM$, we have the following bounds (with $\rho_\cN^{-1}$ denoting the generalized inverse of $\rho_\cN$):\\
\smallskip
\noindent{\it (1)} for $\beta\leq1/2$,
\begin{eqnarray}\label{eq:thm-leq}
&&\frac{\pi}{\sin{\beta\pi}}\|\sigma_\cN^{\beta}\rho_\cN^{-\beta} \rho^{1/2} - 
\sigma^{\beta}\rho^{\frac{1}{2}-\beta}\|_2\\
&\leq& 2\left(\frac{1}{\beta}+
\frac{\|\Delta_{\sigma,\rho}\|}{1-\beta}\right)\frac{1}{T^{\beta}}+T^{\frac{1-2\beta+2\beta^2}{2(1-\beta)}}\,
\left(C_{T, \beta}^f\right)^{1/2}\, (S_f(\rho||\sigma)-S_f(\rho_\cN\|\sigma_\cN))^{1/2}\ ;\nonumber
\end{eqnarray}

\smallskip
\noindent{\it (2)} for $\beta\geq1/2$,  
 \begin{eqnarray}\label{eq:thm-leq3}
&&\frac{\pi}{\sin{\beta\pi}}\|\sigma_\cN^{\beta}\rho_\cN^{-\beta} \rho^{1/2} - 
\sigma^{\beta}\rho^{\frac{1}{2}-\beta}\|_2\\
&\leq& 2\left(\frac{1}{\beta}+
\frac{\|\Delta_{\sigma,\rho}\|}{1-\beta}\right)\frac{1}{T^{1-\beta}}+T^{\beta}\,
\left(C^f_{T, \beta}\right)^{1/2}\, (S_f(\rho||\sigma)-S_f(\rho_\cN\|\sigma_\cN))^{1/2}\ .\nonumber
\end{eqnarray}
\end{theorem}

Note that for $\beta=1/2$ both bounds in the theorem coincide. The proof of the theorem is given in Section~\ref{sec:proof}.

Naturally, we wish to optimize in $T$, and shall do so for specific choices of $f$, so 
that $C_{T, \beta}^f$ is explicit. 
Note that in general, the function $T \mapsto C_{T, \beta}^f$ is, by construction, monotone 
non-decreasing. The right sides of \eqref{eq:thm-leq} and \eqref{eq:thm-leq3} have the form 
\begin{equation}\label{optim}
\phi(T) := K T^{-k} + NT^{n} (C_{T, \beta}^f)^{1/2}
\end{equation}
for $K$, $N$, $k$ and $n > 0$.  
When $C_{T, \beta}^f$ grows like a power of $T$, as it will in examples discussed in the next section, we can 
absorb $(C_{T, \beta}^f)^{1/2}$ into the term $NT^{n}$, and then, after this reduction, 
the optimization is very simple. For later use we record the following simple lemma 
whose proof is elementary calculation.

\begin{lemma}\label{lm:opt}
Let $K, N, k, n >0$. Then the minimum of the function 
$KT^{-k}+NT^n $
on the interval $(0,\infty)$ is 
$$\left(\frac{1}{k}+\frac{1}{n}\right)(kK)^{\frac{n}{k+n}}(nN)^{\frac{k}{k+n}}.$$ 
\end{lemma}
We then have the following Corollary of Theorem~\ref{thm:f-error}:

\begin{corollary} \label{optimcl} Let $f$ be a  regular operator monotone 
decreasing function, and $T>0$.\\
\smallskip
\noindent{\it (1)} For $\beta\leq1/2$, define $T_L(\beta):=T$ and $T_R(\beta):=T^{\beta/(1-\beta)}$.

\smallskip
\noindent{\it (2)} For $\beta\geq1/2$,  define $T_L(\beta):=T^{{(1-\beta)}/{\beta}}$ and $T_R(\beta):=T$.\\
\\
Define
$C_{T, \beta}^f$ to be the least positive constant such that  ${\rm d}t\leq C^f_{T, \beta}{\rm d}\mu_f(t)$ for 
$t\in[T_L(\beta)^{-1}, T_R(\beta)]$. Suppose for some constant $c> 0$, there is a constant $C> 0$ so that
$C^f_{T, \beta} \leq CT^{2c}$. Then there is an explicitly computable  constant 
$K$ depending only on the smallest non-zero eigenvalue of $\rho$, $\beta$, $C$ and $c$, such that, (with $\rho_\cN^{-1}$ denoting the generalized inverse of $\rho_\cN$),  \\
\smallskip
\noindent{\it (1)} for $\beta\leq1/2$,

\begin{equation}\label{optimcl1}
\|\sigma_\cN^{\beta}\rho_\cN^{-\beta} \rho^{1/2} - 
\sigma^{\beta}\rho^{\frac{1}{2}-\beta}\|_2  \leq K 
(S_f(\rho||\sigma)-S_f(\rho_\cN\|\sigma_\cN))^{\frac{\beta(1-\beta)}{1+2c(1-\beta)} }\ ;
\end{equation}

\smallskip
\noindent{\it (2)} for $\beta\geq1/2$, 

\begin{equation}\label{optimcl2}
\|\sigma_\cN^{\beta}\rho_\cN^{-\beta} \rho^{1/2} - 
\sigma^{\beta}\rho^{\frac{1}{2}-\beta}\|_2  \leq K 
(S_f(\rho||\sigma)-S_f(\rho_\cN\|\sigma_\cN))^{\frac12  \frac{1-\beta}{ 1+c} }\ .
\end{equation}
\end{corollary}

\begin{proof}  For $\beta \leq 1/2$, Theorem~\ref{thm:f-error} and the assumption on $C_{T, \beta}^f$
guarantee that
the left side of  \eqref{optimcl1} is bounded by an expression of the form 
$AT^{-a}+BT^b$ where $a= \beta$ and $b = \frac{1 - 2\beta +2\beta^2}{2(1-\beta)} + c$, and where $B$ is a multiple
 of $(S_f(\rho||\sigma)-S_f(\rho_\cN\|\sigma_\cN))^{1/2}$. Then 
$a+b = c+\frac{1}{2(1-\beta)}$, so that $\frac{a}{a+b} = \frac{2\beta(1-\beta)}{1 +2c(1-\beta)}$.
This proves \eqref{optimcl1}. For $\beta \geq 1/2$, we have $a = 1-\beta$ and $b = \beta + c$ so that
$\frac{a}{a+b} = \frac{1-\beta}{1+c}$, and this leads to  \eqref{optimcl2}
\end{proof}

Since 
$\frac{1 - 2\beta +2\beta^2}{2(1-\beta)}= \beta$
at $\beta =1/2$, the two bounds provided by Theorem \ref{thm:f-error}  and by Corollary~\ref{optimcl}  coincide for this value of $\beta$. 
The case in which $\beta = 1/2$ is particularly important. In this case, the quantity 
$\|\sigma_\cN^{\beta}\rho_\cN^{-\beta} \rho^{1/2} - 
\sigma^{\beta}\rho^{\frac{1}{2}-\beta}\|_2$ reduces to 
$\|\sigma_\cN^{1/2}\rho_\cN^{-1/2} \rho^{1/2} - 
\sigma^{1/2}\|_2$. This quantity may be bounded below in terms of the Petz recovery map $\cR_\rho$, defined in \eqref{ac14}.  It was 
 shown in \cite[after Lemma 2.2]{CV17A}  that 
  the following bound holds  
 \begin{equation}\label{recA}
 \|\sR_\rho(\sigma_\cN)-\sigma\|_1  \leq 2\|\sigma_\cN^{1/2}\rho_\cN^{-1/2} \rho^{1/2} - 
 \sigma^{1/2}\|_2. 
 \end{equation}
Exchanging $\rho$ and $\sigma$,
\begin{eqnarray}\label{recB}
\|\sR_\sigma(\rho_\cN)-\rho\|_1  &\leq& 2\|\rho_\cN^{1/2}\sigma_\cN^{-1/2} \sigma^{1/2} - 
 \rho^{1/2}\|_2\nonumber\\
 &\leq&  2\|\rho_\cN^{1/2}\sigma_\cN^{-1/2} (\sigma^{1/2} - 
 \sigma_\cN^{1/2}\rho_\cN^{-1/2}\rho^{1/2})\|_2\nonumber\\
 &\leq& 2\|\rho_\cN\|^{1/2}\|\sigma_\cN^{-1}\|^{1/2}
 \|\sigma_\cN^{1/2}\rho_\cN^{-1/2} \rho^{1/2} - 
 \sigma^{1/2}\|_2  \label{recB0}\\
 &\leq& 2\|\rho\|^{1/2}\|\sigma^{-1}\|^{1/2}
 \|\sigma_\cN^{1/2}\rho_\cN^{-1/2} \rho^{1/2} - 
 \sigma^{1/2}\|_2 .\label{recB}
\end{eqnarray}
The last inequality follows from the fact that the spectra of $\sigma_\cN$ and $\rho_\cN$ lie in the convex hulls of the spectra of $\sigma$ and $\rho$ respectively, and it might even be the case that $\sigma_\cN$ is invertible when $\sigma$ is not. Thus the inequality \eqref{recB} is weaker than \eqref{recB0}, though it may be useful to have a bound in terms of
$\rho$ and $\sigma$ themselves. 

This brings us to our second corollary:

\begin{corollary} \label{optimclB} Let $\cN$ be a von Neumann sub-algebra of a finite dimensional von Neumann algebra $\cM$. 
Let $f$ be a  regular operator monotone 
decreasing function, $T > 0$, define
$C_{T, 1/2}^f$ to be the least positive constant such that  ${\rm d}t\leq C^f_{T, 1/2}{\rm d}\mu_f(t)$ for 
$t\in[T^{-1},T]$. Suppose for some constant $c> 0$, there is a constant $C> 0$ so that
$C^f_{T, 1/2} \leq CT^{2c}$. Then for all density matrices $\rho$, $\sigma$ with $\sigma_\cN$ invertible,  
there is an explicitly computable  constant 
$K$ depending only on  the smallest non-zero eigenvalue of $\rho$, $\|\sigma_\cN^{-1}\|$, $\beta$, $C$ and $c$, such that 
\begin{equation}\label{optimclB2}
\max\{ \|\sR_\rho(\sigma_\cN)-\sigma\|_1 \ ,\ \|\sR_\sigma(\rho_\cN)-\rho\|_1\}   \leq K 
(S_f(\rho||\sigma)-S_f(\rho_\cN\|\sigma_\cN))^{\frac14  \frac{1}{ 1+c} }\ .
\end{equation}
Dropping the assumption that $\sigma_\cN$ is invertible, the bound becomes
\begin{equation}\label{optimclB2A}
 \|\sR_\rho(\sigma_\cN)-\sigma\|_1    \leq K 
(S_f(\rho||\sigma)-S_f(\rho_\cN\|\sigma_\cN))^{\frac14  \frac{1}{ 1+c} }\ .
\end{equation}
where now $K$ depends only on the smallest non-zero eigenvalue of $\rho$, $\beta$, $C$ and $c$.
\end{corollary}

Consequently, $S_f(\rho||\sigma) = S_f(\rho_\cN\|\sigma_\cN)$ if and only if both 
$\sigma = \sR_\rho(\sigma_\cN)$ and  $\rho = \sR_\sigma(\rho_\cN)$.  Conversely, if either of these equations is valid, say $\sigma = \sR_\rho(\sigma_\cN)$, then by 
 Petz's monotonicity theorem, 
$$S_f( \rho||\sigma) = S_f(\sR_\rho(\rho_\cN)||\sR_\rho(\sigma_\cN)) \leq  
S_f(\rho_\cN||\sigma_\cN) \leq S_f(\rho||\sigma)\ ,$$
and then $S_f(\rho||\sigma) = S_f(\rho_\cN)||\sigma_\cN)$, implying that the other equation, 
$\rho = \sR_\sigma(\rho_\cN)$, is also valid. This symmetry may be seen from a detailed analysis of the solution set of the Petz equation $\gamma = \sR_\sigma(\gamma_\cN)$; see for example \cite{CV17A}, and it was proved by Petz \cite{P86} through another more complicated argument. However, it is worth noting that this symmetry, valid when $\sigma_\cN$ is invertible,  is an immediate consequence of Corollary~\ref{optimclB}.

We obtain other interesting information for values of $\beta$ other than $\beta = 1/2$. Reasoning as above, note that 
\begin{equation}\label{otherbeta}
\|\sigma_\cN^{\beta}\rho_\cN^{-\beta}  - 
\sigma^{\beta}\rho^{-\beta}\|_2  = \|(\sigma_\cN^{\beta}\rho_\cN^{-\beta} \rho^{1/2} - 
\sigma^{\beta}\rho^{\frac{1}{2}-\beta})\rho^{-1/2}\|_2
\leq \|\rho^{-1/2}\| \|\sigma_\cN^{\beta}\rho_\cN^{-\beta} \rho^{1/2} - 
\sigma^{\beta}\rho^{\frac{1}{2}-\beta}\|_2
\end{equation}
Thus absorbing the factor of $\|\rho^{-1/2}\|$, into the constant $K$, Corollary~\ref{optimcl} can be restated with
$\|\sigma_\cN^{\beta}\rho_\cN^{-\beta}  - 
\sigma^{\beta}\rho^{-\beta}\|_2$ on the left sides of \eqref{optimcl1} and  \eqref{optimcl2} in place of 
$\|\sigma_\cN^{\beta}\rho_\cN^{-\beta} \rho^{1/2} - 
\sigma^{\beta}\rho^{\frac{1}{2}-\beta}\|_2$. 
We then conclude, arguing as above, $S_f(\rho||\sigma) - S_f(\rho_\cN||\sigma_\cN) = 0$ if and only if 
$$(\sigma_\cN)^{\beta}(\rho_\cN)^{-\beta} =\sigma^{\beta}\rho^{-\beta} $$
for all $\beta\in (0,1)$, and then, since for any positive matrix $X$, $\beta \mapsto X^\beta$ is an entire analytic function, this identity holds for all $\beta\in \C$.

\section{Examples}\label{sec:ex}

In this section we apply the previous result, Theorem \ref{thm:f-error}, to two particular cases: the logarithmic and the power functions. 

\subsection{Logarithmic function}

In the previous section, let us take $f(x)=-\log(x)$, then from Example \ref{logex} we have ${\rm d}\mu_{\log}(t)={\rm d}t,$ and $C_{T,\beta}^f=1$. The corresponding quasi-relative entropy is the Umegaki relative entropy (\ref{Um}).

\begin{corollary}\label{cor:log}
For $\beta<1/2$, the Umegaki relative entropy satisfies                        
\begin{eqnarray}\label{eq:log2}
S(\rho||\sigma)-S(\rho_\cN\|\sigma_\cN)&\geq&  K^L_{\beta}(\rho,\sigma)\|(\sigma_\cN)^{\beta}(\rho_\cN)^{-\beta} \rho^{1/2} - \sigma^{\beta}\rho^{\frac{1}{2}-\beta}\|_2^\frac{1}{\beta(1-\beta)}\ ,
\end{eqnarray}
where $K^L_{\beta}(\rho,\sigma)= \left(\frac{\pi(1-2\beta+2\beta^2)\beta}{\sin{\beta\pi}}\right)^\frac{1}{\beta(1-\beta)}\left(1+\frac{\beta}{1-\beta}\|\Delta_{\sigma,\rho}\| \right)^{-\frac{1-2\beta+2\beta^2}{\beta(1-\beta)}}2^{-\frac{1-2\beta+2\beta^2}{\beta(1-\beta)}}\left(\frac{1-2\beta+2\beta^2}{2(1-\beta)}\right)^{-2}$.

For $\beta\geq1/2$, the Umegaki relative entropy satisfies
\begin{eqnarray}\label{eq:log1}
&&S(\rho||\sigma)-S(\rho_\cN\|\sigma_\cN)
\geq K^U_{\beta}(\rho,\sigma) \|(\sigma_\cN)^{\beta}(\rho_\cN)^{-\beta} \rho^{1/2} - \sigma^{\beta}\rho^{\frac{1}{2}-\beta}\|_2^{\frac{2}{1-\beta}}\ ,
\end{eqnarray}
where $K^U_{\beta}(\rho,\sigma)= \left(\frac{\pi\beta(1-\beta)}{\sin{\beta\pi}}\right)^\frac{2}{1-\beta}\left(\frac{1-\beta}{\beta}+\|\Delta_{\sigma,\rho}\| \right)^{-\frac{2\beta}{1-\beta}}2^{-\frac{2\beta}{1-\beta}}\beta^{-2}$.

In particular, taking $\beta=1/2$, we obtain an inequality very close to \cite{CV17A}
\begin{eqnarray}\label{eq:log3}
&&S(\rho||\sigma)-S(\rho_\cN\|\sigma_\cN)
\geq \left(\frac{\pi}{4}\right)^4(1+\|\Delta_{\sigma,\rho}\| )^{-2} \|(\sigma_\cN)^{1/2}(\rho_\cN)^{-1/2} \rho^{1/2} - \sigma^{1/2}\|_2^4\ .
\end{eqnarray}

\end{corollary}

\begin{proof} The proof follows directly from Theorem \ref{thm:f-error} and Lemma \ref{lm:opt}.
\end{proof}

\subsection{Power function}

Another interesting example of $f$-relative quasi-entropy is given by the power function. These types of  quasi-entropies appear in the definition of the R\'enyi entropy, which will be discussed in the next section. Let  $p_\alpha(t)=-t^\alpha$ for $\alpha\in(0,1)$.  From Example \ref{powex} we have that  ${\rm d}\mu_{p_\alpha}(x) = \pi^{-1}\sin(\alpha \pi) x^\alpha{\rm d}x$. Therefore, the power-relative quasi-entropy $S_{p_\alpha}=-\Tr(\rho^{1-\alpha}\sigma^{\alpha})$ satisfies the following corollary.

\begin{corollary}\label{cor:power} For $\alpha\in(0,1)$ the power-relative quasi-entropy satisfies:
for $\beta<1/2$, 
\begin{eqnarray}\label{eq:power2}
S_{p_\alpha}(\rho||\sigma)-S_{p_\alpha}(\rho_\cN\|\sigma_\cN)&\geq& K^L_{\alpha,\beta}(\rho,\sigma) \|(\sigma_\cN)^{\beta}(\rho_\cN)^{-\beta} \rho^{1/2} - \sigma^{\beta}\rho^{\frac{1}{2}-\beta}\|_2^\frac{4-2\beta+\alpha(1-\beta)}{1-\beta^2}  \ ,
\end{eqnarray}
where $K^L_{\alpha,\beta}(\rho,\sigma)$  is defined in \eqref{eq:pow-KL} below;

for $\beta\geq1/2$, 
\begin{eqnarray}\label{eq:power1}
&&S_{p_\alpha}(\rho||\sigma)-S_{p_\alpha}(\rho_\cN\|\sigma_\cN)\geq K^U_{\alpha,\beta}(\rho,\sigma) \|(\sigma_\cN)^{\beta}(\rho_\cN)^{-\beta} \rho^{1/2} - \sigma^{\beta}\rho^{\frac{1}{2}-\beta}\|_2^\frac{2(1+\beta)+\alpha(1-\beta)}{1-\beta^2}\ ,
\end{eqnarray}
where $ K^U_{\alpha,\beta}(\rho,\sigma)$ is defined in \eqref{eq:pow-KU} below.

In particular, for $\beta=1/2$
\begin{eqnarray}\label{eq:power3}
S_{p_\alpha}(\rho||\sigma)-S_{p_\alpha}(\rho_\cN\|\sigma_\cN)
&\geq& K^U_{\alpha}(\rho,\sigma) \|(\sigma_\cN)^{1/2}(\rho_\cN)^{-1/2} \rho^{1/2} - \sigma^{1/2}\|_2^{4+{2}\alpha}\ .
\end{eqnarray}
Here $K^U_{\alpha}(\rho,\sigma)=K^U_{\alpha,1/2}(\rho,\sigma)=K^L_{\alpha,1/2}(\rho,\sigma)$.
\end{corollary}

\begin{proof}
For $\beta<1/2$, $C_{T,\beta}^{p_\alpha}\leq \frac{\pi}{\sin(\alpha\pi)}T^{\alpha}.$ Then from Theorem \ref{thm:f-error} we have 
\begin{eqnarray*}
&&\frac{\pi}{\sin\beta\pi}\|(\sigma_\cN)^{\beta}(\rho_\cN)^{-\beta} \rho^{1/2} - \sigma^{\beta}\rho^{\frac{1}{2}-\beta}\|_2\\
&\leq& 2\left(\frac{1}{\beta}+
\frac{\|\Delta_{\sigma,\rho}\|}{1-\beta}\right)\frac{1}{T^{\beta}}+T^{\frac{1-2\beta+2\beta^2}{2(1-\beta)}+\frac{\alpha}{2}}\,
\left(\frac{\pi}{\sin(\alpha\pi)}\right)^{1/2}\, (S_f(\rho||\sigma)-S_f(\rho_\cN\|\sigma_\cN))^{1/2}\ .
\end{eqnarray*}
Using Lemma \ref{lm:opt} we obtain
\begin{eqnarray*}\
&&S_{p_\alpha}(\rho||\sigma)-S_{p_\alpha}(\rho_\cN\|\sigma_\cN)
\geq K^L_{\alpha,\beta}(\rho,\sigma) \|(\sigma_\cN)^{\beta}(\rho_\cN)^{-\beta} \rho^{1/2} - \sigma^{\beta}\rho^{\frac{1}{2}-\beta}\|_2^\frac{1+\alpha(1-\beta)}{\beta(1-\beta)}\ ,
\end{eqnarray*}
for \begin{eqnarray}\label{eq:pow-KL}
&&K^L_{\alpha,\beta}(\rho,\sigma)= (1+\frac{\beta}{1-\beta}\|\Delta_{\sigma,\rho}\|)^{-\frac{\alpha(1-\beta)+1-2\beta+2\beta^2}{\beta(1-\beta)}}2^{-\frac{\alpha(1-\beta)+1-2\beta+2\beta^2}{\beta(1-\beta)}} \\
&&\cdot\frac{\sin(\alpha\pi)}{\pi}\left( \frac{\pi\beta(\alpha(1-\beta)+1-2\beta+2\beta^2)}{(1+\alpha(1-\beta))\sin\beta\pi}\right)^{\frac{1+\alpha(1-\beta)}{\beta(1-\beta)}}\left(\frac{\alpha(1-\beta)+1-2\beta+2\beta^2}{2(1-\beta)}\right)^{-2} \ .\nonumber
\end{eqnarray}

For $\beta\geq1/2$, we have $C_{T,\beta}^{p_\alpha}\leq \frac{\pi}{\sin(\alpha\pi)}T^{\alpha(1-\beta)/\beta}.$ Therefore,
\begin{eqnarray*}
&&\frac{\pi}{\sin\beta\pi}\|(\sigma_\cN)^{\beta}(\rho_\cN)^{-\beta} \rho^{1/2} - \sigma^{\beta}\rho^{\frac{1}{2}-\beta}\|_2\\
&\leq&2\left(\frac{1}{\beta}+\frac{\|\Delta_{\sigma,\rho}\|}{1-\beta}\right)\frac{1}{T^{1-\beta}}+T^{\beta+\frac{\alpha(1-\beta)}{2\beta}}\,
\left(\frac{\pi}{\sin(\alpha\pi)}\right)^{1/2}\,  (S_{p_\alpha}(\rho||\sigma)-S_{p_\alpha}(\rho_\cN\|\sigma_\cN))^{1/2}\ .
\end{eqnarray*}
Similarly to the previous case, using Lemma \ref{lm:opt} we obtain
\begin{eqnarray*}
&&S_{p_\alpha}(\rho||\sigma)-S_{p_\alpha}(\rho_\cN\|\sigma_\cN)
\geq K^U_{\alpha,\beta}(\rho,\sigma) \|(\sigma_\cN)^{\beta}(\rho_\cN)^{-\beta} \rho^{1/2} - \sigma^{\beta}\rho^{\frac{1}{2}-\beta}\|_2^\frac{2\beta+\alpha(1-\beta)}{\beta(1-\beta)}\ ,
\end{eqnarray*}
for  
\begin{eqnarray}\label{eq:pow-KU}
&&K^U_{\alpha,\beta}(\rho,\sigma)=\left(\frac{1-\beta}{\beta}+\|\Delta_{\sigma,\rho}\| \right)^{-\frac{2\beta^2+\alpha(1-\beta)}{\beta(1-\beta)}}2^{-\frac{2\beta^2+\alpha(1-\beta)}{\beta(1-\beta)}}\\
&&\cdot\frac{\sin\alpha\pi}{\pi}\left(\frac{\pi(1-\beta)(2\beta^2+\alpha(1-\beta))}{(2\beta+\alpha(1-\beta))\sin\beta\pi} \right)^{\frac{2\beta+\alpha(1-\beta)}{\beta(1-\beta)}}\left(\frac{2\beta^2+\alpha(1-\beta)}{2\beta} \right)^{-2}  \ .\nonumber
\end{eqnarray}

Plugging in $\beta=1/2$ in the last inequality we obtain the bound claimed in the corollary. 
\end{proof}

\section{R\'enyi entropy}

For $\alpha\in(0,1)$ R\'enyi entropy is be defined  as
$$S_\alpha(\rho\|\sigma):=\frac{1}{\alpha-1}\log\Tr(\rho^\alpha\sigma^{1-\alpha})  = 
\frac{1}{\alpha-1} \log( -S_{p_{1-\alpha}}(\rho\|\sigma)) \ , $$
where $S_{p_{1-\alpha}}(\rho\|\sigma)$ is the power quasi entropy of the previous section with power $1-\alpha$.
Notice that for all $\alpha\in (0,1)$, $\rho$ and $\sigma$, 
\begin{equation}\label{sbounds}
0 < -S_{p_{1-\alpha}}(\rho\|\sigma)  \leq 1
\end{equation}
where the later inequality follows from H\"older's inequality.

\begin{theorem}\label{thm:renyi}
For any $\alpha\in(0,1)$,  any density matrix  $\rho$ and any non-degenerate density matrix  $\sigma$,
\begin{eqnarray}\label{thm:renyiA}
S_\alpha(\rho||\sigma)-S_\alpha(\rho_\cN||\sigma_\cN)\geq
\frac{1}{1-\alpha}\log\left( 1 +  {K}^U_{1-\alpha}(\rho,\sigma) \| (\sigma_\cN)^{1/2}(\rho_\cN)^{-1/2} \rho^{1/2} -
 \sigma^{1/2}\|_2^{6-2\alpha}{}\right)\ .
\end{eqnarray}
where ${K}^U_{1-\alpha}(\rho,\sigma) $ is given in Corollary \ref{cor:power}. 
\end{theorem}

\begin{corollary}\label{remax}
For any $\alpha\in(0,1)$, suppose $\rho$ and $\sigma$ are such that 
Corollary \ref{optimclB} is satisfied. Then
\begin{multline}\label{thm:renyi2}
S_\alpha(\rho||\sigma)-S_\alpha(\rho_\cN||\sigma_\cN)\geq\\
\frac{1}{1-\alpha} \log\left(1 + \frac12\|\rho\|^{-1/2}\|\sigma^{-1}\|^{-1/2} K^U_{1-\alpha}(\rho,\sigma) \max\{\| \sR_\rho(\sigma_\cN)-\sigma\|_1, \| \sR_\sigma(\rho_\cN)-\rho\|_1\}^{6-2\alpha}\right)\ ,
\end{multline}
where ${K}^U_{1-\alpha}(\rho,\sigma) $ is given in Corollary \ref{cor:power}. 
\end{corollary}

\begin{remark} An equivalent formulation of \eqref{thm:renyi2} is of course that
\begin{multline}
\max\{\| \sR_\rho(\sigma_\cN)-\sigma\|_1, \| \sR_\sigma(\rho_\cN)-\rho\|_1\}^{6-2\alpha}  \leq \\
\frac{ 2\|\rho\|^{1/2}\|\sigma^{-1}\|^{1/2} }{ K^U_{1-\alpha}(\rho,\sigma) }\left(\exp[ 
(1-\alpha)(S_\alpha(\rho||\sigma)-S_\alpha(\rho_\cN||\sigma_\cN))]-1\right)
\end{multline}
We are of course most interested in the case in which $S_\alpha(\rho||\sigma)- S_\alpha(\rho_\cN||\sigma_\cN)$ is small,   Since for any $r>0$,
$e^{(1-\alpha)x} -1 \leq (x/r)(e^{(1-\alpha)r} -1)$, one can simplify the bound if one knows that 
$S_\alpha(\rho||\sigma)- S_\alpha(\rho_\cN||\sigma_\cN)$ is small.
\end{remark}

\begin{proof}[Proof of Corollary~\ref{remax}] This is immediate from Theorem~\ref{thm:renyi} and \eqref{recA} and \eqref{recB}.
\end{proof}

\begin{proof}[Proof of Theorem~\ref{thm:renyi}]
\begin{align*}
S_\alpha(\rho||\sigma)-S_\alpha(\rho_\cN||\sigma_\cN)&= \frac{1}{\alpha-1}\log\frac{S_{p_{1-\alpha}}(\rho||\sigma)}{S_{p_{1-\alpha}}(\rho_\cN||\sigma_\cN)}\\
&=\frac{1}{1-\alpha}\log\frac{S_{p_{1-\alpha}}(\rho_\cN||\sigma_\cN)}{S_{p_{1-\alpha}}(\rho||\sigma)}\\
&=\frac{1}{1-\alpha}\log\left(1+\frac{S_{p_{1-\alpha}}(\rho||\sigma)-
S_{p_{1-\alpha}}(\rho_\cN||\sigma_\cN)}{-S_{p_{1-\alpha}}(\rho||\sigma)}\right) \\
&\geq \frac{1}{1-\alpha}\log\left(1+ S_{p_{1-\alpha}}(\rho||\sigma)-
S_{p_{1-\alpha}}(\rho_\cN||\sigma_\cN)\right)
\ ,\end{align*}
where in the last line we have used \eqref{sbounds} and the monotonicity of the logarithm. Making one more use of this monotonicity,
we now simply apply the  bound \eqref{eq:power3} from the previous section with $1-\alpha$ in place of $\alpha$. 
\end{proof}

Note that at the end of the proof of this theorem, we could use Corollary (\ref{cor:power}) for any $ \beta$ instead of $1/2$ to obtain a stronger lower bound as a one-parameter family.

\section{Proof of Theorem \ref{thm:f-error}}\label{sec:proof}

From the definition of the quantum $f$-relative entropy (\ref{Petsqdef}), we construct the following family of relative entropies: for each $t>0$, the function $x\mapsto (t+x)^{-1}$ is operator convex, and so define a one parameter family a quasi relative entropies by
\begin{equation}\label{sfrel}
S_{(t)}(\rho||\sigma) =  \Tr\left[ (t + \Delta_{\sigma,\rho})^{-1}\rho\right] \ .
\end{equation}
For an operator monotone decreasing function $f$ (which implies that $f$ is operator convex), according to the integral representation (\ref{low}) the $f$-relative quasi-entropy $S_{f}$ can we written as
$$S_f(\rho\|\sigma)=-a-b+\int_0^\infty\left(S_{(t)}(\rho||\sigma)-\frac{t}{t^2+1}\right){\rm d}\mu_f(t)\ ,  $$
for $a\geq0$ and $b\in\bR$. 

For an orthogonal projection $\sE_\tau$ from the von Neumann algebra $\cM$ to the sub-algebra $\cN$ denote the processed states as $\rho_\cN := \sE_\tau \rho$ and $\sigma_\cN := \sE_\tau \sigma$. The monotonicity inequality (\ref{lindmon}) holds, and we are interested in the lower bound on the relative entropy difference. From the above integral representation, it is clear that the difference between relative entropies can be written in terms of the  $S_{(t)}$-family, 
$$S_f(\rho\|\sigma)-S_f(\rho_\cN\|\sigma_\cN)=\int_0^\infty\left(S_{(t)}(\rho||\sigma)-S_{(t)}(\rho_\cN||\sigma_\cN)\right){\rm d}\mu_f(t)\ .$$

\begin{proof}
From \cite[Lemma 2.1]{CV17A}  we have
$$S_{(t)}(\rho||\sigma) - S_{(t)}(\rho_\cN||\sigma_\cN) \geq  t\|w_t\|_2^2, $$
where
$$w_t := U (t\one +\Delta_{\sigma_\cN,\rho_\cN})^{-1} (\rho_\cN)^{1/2} -  (t\one  + \Delta_{\sigma,\rho})^{-1}\rho^{1/2}\ , $$
with  the operator $U$ being the mapping $\cH:=(\cM, \langle \cdot,\cdot\rangle_{HS})$ to itself defined as
\begin{equation*}\label{Ujdef}
U(X) = \sE_\tau(X)\rho_\cN^{-1/2}\rho^{1/2}\ .
\end{equation*}
Notice that for all $X\in\cN$, $U(X)=X\rho_\cN^{-1/2}\rho^{1/2}.$ On account of this identity, 
$$-w_t = U[t^{-1}\one  -  (t\one  +\Delta_{\sigma_\cN,\rho_\cN})^{-1}] (\rho_\cN)^{1/2} -
 [t^{-1}\one  -  (t\one + \Delta_{\sigma,\rho})^{-1}]\rho^{1/2}\ . $$
Therefore, since $U$ is a contraction, 
$$\|w_t\| \leq  \| [t^{-1}\one -  (t\one +\Delta_{\sigma_\cN,\rho_\cN})^{-1}] (\rho_\cN)^{1/2} \| + 
 \| [t^{-1}\one -  (t\one + \Delta_{\sigma,\rho})^{-1}]\rho^{1/2}\|\ .$$
Since on account of the non-negativity of the modular operator,  $0 \leq t^{-1}\one -  (t\one +\Delta_{\sigma_\cN,\rho_\cN})^{-1}  \leq t^{-1}\one$,
with the analogous estimate valid with $\Delta_{\sigma,\rho}$ in place of $ \Delta_{\sigma_\cN,\rho_\cN}$, Therefore,
\begin{equation}\label{newtwist}
\|w_t\| \leq 2 t^{-1}\ .
\end{equation}

 Now using the integral representation of the power function from Example \ref{powex} (recall that $\beta\in(0,1)$)
 $$
 X^{\beta} = \frac{\sin{\beta\pi}}{\pi} \int_0^\infty t^{\beta} \left(\frac{1}{t}\one  - \frac{1}{t+X}\right){\rm d}t,
 $$
 and $U(\rho_\cN)^{1/2} = \rho^{1/2}$, we conclude that
 \begin{equation}\label{eq:U-int}
 U(\Delta_{\sigma_\cN,\rho_\cN})^{\beta} (\rho_\cN)^{1/2}  -  (\Delta_{\sigma,\rho})^{\beta}\rho^{1/2} = -  \frac{\sin{\beta\pi}}{\pi}\int_0^\infty t^{\beta}w_t{\rm d}t\ .
\end{equation}
 On the other hand,
 \begin{eqnarray}\label{eq:U-states}
 U(\Delta_{\sigma_\cN,\rho_\cN})^{\beta} (\rho_\cN)^{1/2}  -  (\Delta_{\sigma,\rho})^{\beta}\rho^{1/2} &=& \sigma_\cN^{\beta}\rho_\cN^{-\beta} \rho^{1/2} - \sigma^{\beta}\rho^{\frac{1}{2}-\beta}\ .
 \end{eqnarray}
 
 Combining the last two equalities (\ref{eq:U-int}) and (\ref{eq:U-states}), and taking the Hilbert space norm associated with $\cH$, for any $T_L, T_R>0$, 
\begin{eqnarray}\label{eq:REM1}
&&\|(\sigma_\cN)^{\beta}(\rho_\cN)^{-\beta} \rho^{1/2} - \sigma^{\beta}\rho^{\frac{1}{2}-\beta}\|_2 =
\frac{\sin{\beta\pi}}{\pi}\left\Vert \int_0^\infty  t^{\beta} w_t{\rm d}t\right\Vert_2  \nonumber\\
&\leq& \frac{\sin{\beta\pi}}{\pi}\int_0^{1/T_L} t^{\beta}  \|w_t\|_2{\rm d}t + \frac{\sin{\beta\pi}}{\pi} \int_{1/T_L}^{T_R} t^{\beta}  \|w_t\|_2{\rm d}t + \frac{\sin{\beta\pi}}{\pi}\left\| \int_{T_R}^\infty t^{\beta} w_t {\rm d}t \right\|_2\ .
\end{eqnarray}
Let us look at these three terms separately. The first term can be bounded using \eqref{newtwist}:
\begin{eqnarray}\label{eq:REM1-1}
&&\int_0^{1/T_L} t^{\beta}  \|w_t\|_2{\rm d}t \leq 2 \int_0^{1/T_L} t^{\beta -1}  {\rm d}t=\frac{2}{\beta}\frac{1}{T_L^{\beta}}\ .
\end{eqnarray}

The third term in (\ref{eq:REM1}) can be bounded the following way: For any positive  operator $X> 0$,
$$
t^{\beta} \left(\frac{1}{t}\one  - \frac{1}{t+X}\right)   
\leq  t^{\beta} \left(\frac{1}{t} - \frac{1}{t+\|X\|}\right)\one  =  \frac{t^{\beta-1}\|X\|}{(t+\|X\|)}\one,
$$
and hence
$$
\int_{T}^\infty t^{\beta} \left(\frac{1}{t}\one - \frac{1}{t+X}\right) {\rm d}t \leq \|X\|^{\beta}\left(\int_{T/\|X\|}^\infty \frac{t^{\beta-1}}{1+t}  {\rm d}t\right) \one \leq \frac{\|X\|}{{(1-\beta)}T^{1-\beta}}\one\ .
$$
Since spectra of $\sigma_\cN$ and $\rho_\cN$ lie in the convex hulls of the spectra of $\sigma$ and $\rho$ respectively, it follows that $\|\Delta_{\sigma_\cN,\rho_\cN}\|  \leq \|\Delta_{\sigma,\rho}\|$.  Therefore, recalling the definition of $w_t$, we obtain
\begin{equation}\label{eq:third}
\left\| \int_{T_R}^\infty t^{\beta} w_t {\rm d}t \right\|_2  \leq \frac{2 \|\Delta_{\sigma,\rho}\|}{(1-\beta) \, T_R^{1-\beta}}\ .
\end{equation}

The second term can be bounded using Cauchy-Schwartz inequality and the equivalence of measures on the finite interval, i.e. there is a constant $C^f_{T_L, T_R}$ such that $dt\leq C^f_{T_L, T_R}{\rm d}\mu_f(t)$ 
 for $t\in[1/T_L, T_R]$.  {\bf Case 1:} $\beta\leq1/2$.
\begin{eqnarray}\label{eq:REM2-1}
\left(\int_{1/T_L}^{T_R} t^{\beta}  \|w_t\|_2{\rm d}t\right)^2&\leq& T_R\int_{1/T_L}^{T_R} t^{2\beta}  \|w_t\|_2^2{\rm d}t\nonumber\\
&\leq& T_RT_L^{1-2\beta}\int_{1/T_L}^{T_R} t \|w_t\|_2^2{\rm d}t\nonumber\\
&\leq&  T_RT_L^{1-2\beta}\int_{1/T_L}^{T_R} S_{(t)}(\rho||\sigma)-S_{(t)}(\rho_\cN\|\sigma_\cN) {\rm d}t\nonumber\\
&\leq&  T_RT_L^{1-2\beta}C^f_{T_L, T_R}\int_{1/T_L}^{T_R} S_{(t)}(\rho||\sigma)-S_{(t)}(\rho_\cN\|\sigma_\cN) {\rm d}\mu_f(t)\nonumber\\
&\leq& T_RT_L^{1-2\beta}\,C^f_{T_L, T_R}\, (S_f(\rho||\sigma)-S_f(\rho_\cN\|\sigma_\cN))\ .
\end{eqnarray}

 Therefore, combining (\ref{eq:REM1}), (\ref{eq:REM1-1}), (\ref{eq:third}), and (\ref{eq:REM2-1})  we have 
\begin{eqnarray*}\label{eq:bound}
\frac{\pi}{\sin{\beta\pi}}\|(\sigma_\cN)^{\beta}(\rho_\cN)^{-\beta} \rho^{1/2} - \sigma^{\beta}\rho^{\frac{1}{2}-\beta}\|_2
&\leq& \frac{2}{\beta T_L^{\beta}}+\frac{2 \|\Delta_{\sigma,\rho}\|}{(1-\beta) \, T_R^{1-\beta}}\\&&+T_R^{1/2}T_L^{1/2-\beta}\,\left(C^f_{T_L, T_R}\right)^{1/2}\, (S_f(\rho||\sigma)-S_f(\rho_\cN\|\sigma_\cN))^{1/2}\ .
\end{eqnarray*}

Taking $T_L:=T$ and  $T_R:=T^{\beta/(1-\beta)}$ we obtain
\begin{eqnarray*}\label{eq:bound2}
&&\frac{\pi}{\sin{\beta\pi}}\|(\sigma_\cN)^{\beta}(\rho_\cN)^{-\beta} \rho^{1/2} - \sigma^{\beta}\rho^{\frac{1}{2}-\beta}\|_2\\
&\leq& 2\left(\frac{1}{\beta}+\frac{\|\Delta_{\sigma,\rho}\|}{1-\beta}\right)\frac{1}{T^{\beta}}+T^{\frac{1-\beta}{2}+\frac{\beta^2}{2(1-\beta)}}\,\left(C_{T, \beta}^f\right)^{1/2}\, (S_f(\rho||\sigma)-S_f(\rho_\cN\|\sigma_\cN))^{1/2}\ .
\end{eqnarray*}

{\bf Case 2:} $\beta>1/2$. 
\begin{eqnarray}\label{eq:REM1-2}
\left(\int_{1/T_L}^{T_R} t^{\beta}  \|w_t\|_2{\rm d}t\right)^2&\leq& T_R\int_{1/T_L}^{T_R} t^{2\beta}  \|w_t\|_2^2{\rm d}t\nonumber\\
&\leq& T_R^{2\beta}\int_{1/T_L}^{T_R} t \|w_t\|_2^2{\rm d}t\nonumber\\
&\leq&  T_R^{2\beta}\int_{1/T_L}^{T_R}( S_{(t)}(\rho||\sigma)-S_{(t)}(\rho_\cN\|\sigma_\cN) ){\rm d}t\nonumber\\
&\leq&  T_R^{2\beta}C^f_{T_L, T_R}\int_{1/T_L}^{T} (S_{(t)}(\rho||\sigma)-S_{(t)}(\rho_\cN\|\sigma_\cN) ){\rm d}\mu_f(t)\nonumber\\
&\leq& T_R^{2\beta}\,C^f_{T_L, T_R}\, (S_f(\rho||\sigma)-S_f(\rho_\cN\|\sigma_\cN))\ .
\end{eqnarray}

Therefore, combining (\ref{eq:REM1}), (\ref{eq:REM1-1}), (\ref{eq:third}), and (\ref{eq:REM1-2}) we have
\begin{eqnarray*}\label{eq:bound}
\frac{\pi}{\sin{\beta\pi}} \|(\sigma_\cN)^{\beta}(\rho_\cN)^{-\beta} \rho^{1/2} - \sigma^{\beta}\rho^{\frac{1}{2}-\beta}\|_2&\leq& \frac{2}{\beta T_L^{\beta}}+\frac{2 \|\Delta_{\sigma,\rho}\|}{(1-\beta) \, T_R^{1-\beta}}\\&&+T_R^{\beta}\,\left(C^f_{T, \beta}\right)^{1/2}\, (S_f(\rho||\sigma)-S_f(\rho_\cN\|\sigma_\cN))^{1/2}\ .
\end{eqnarray*}

Taking $T_L:=T^{(1-\beta)/\beta}$ and  $T_R:=T$ we obtain
\begin{eqnarray*}\label{eq:bound2}
&&\frac{\pi}{\sin{\beta\pi}}\|(\sigma_\cN)^{\beta}(\rho_\cN)^{-\beta} \rho^{1/2} - \sigma^{\beta}\rho^{\frac{1}{2}-\beta}\|_2\\
&\leq& 2\left(\frac{1}{\beta}+\frac{\|\Delta_{\sigma,\rho}\|}{1-\beta}\right)\frac{1}{T^{1-\beta}}+T^{\beta}\,\left(C^f_{T, \beta}\right)^{1/2}\, (S_f(\rho||\sigma)-S_f(\rho_\cN\|\sigma_\cN))^{1/2}\ .
\end{eqnarray*}

\end{proof}

\vspace{0.3in}
\textbf{Acknowledgments.} EAC was partially supported by NSF grant DMS  1501007. AV is grateful to EAC for hosting her visit to Rutgers University, during which this work was partially completed.  We thank the referees for raising questions that have led us 
to  improve the results and their presentation.

\end{document}